\documentclass[11pt]{article}
\usepackage[utf8]{inputenc}
\usepackage[T1]{fontenc}
\usepackage{amsmath,amssymb,amsthm}
\usepackage{mathtools}
\usepackage{booktabs}
\usepackage{longtable}
\usepackage{enumitem}
\usepackage{tikz}
\usetikzlibrary{arrows.meta,positioning,shapes.geometric}
\usepackage[hidelinks]{hyperref}
\newcommand{\doi}[1]{\href{https://doi.org/#1}{doi:\detokenize{#1}}}
\newcommand{\ver}{\mathsf{ver}}

\hypersetup{
  pdftitle={CA-less Mutual Co-Signing over a Unidirectional Visual Channel with Transported Hardware Attestation},
  pdfauthor={Dmytro Diikun},
  pdfsubject={Cryptographic protocols; digital signatures; hardware attestation; offline co-signing},
  pdfkeywords={digital signatures, hardware attestation, certificate-less trust, offline protocols, fountain codes, secure element, signature chaining},
  pdfcreator={}, pdfproducer={}
}

\theoremstyle{definition}
\newtheorem{definition}{Definition}
\newtheorem{remark}{Remark}
\theoremstyle{plain}
\newtheorem{theorem}{Theorem}
\newtheorem{proposition}{Proposition}
\newtheorem{lemma}{Lemma}

\newcommand{\RS}{\langle\mathsf{RS}\rangle}
\newcommand{\US}{\langle\mathsf{US}\rangle}
\newcommand{\cat}{\,\|\,}
\newcommand{\enc}{\mathsf{enc}}
\newcommand{\Sign}{\mathsf{Sign}}
\newcommand{\Vrfy}{\mathsf{Vrfy}}
\newcommand{\KGen}{\mathsf{KGen}}
\newcommand{\Adv}{\mathrm{Adv}}
\newcommand{\Att}{\mathsf{Att}}
\newcommand{\AttVer}{\mathsf{AttVer}}
\newcommand{\Gen}{\mathsf{Gen}}

\title{CA-less Mutual Co-Signing of Documents over a\\
Unidirectional Visual Channel with Transported Hardware Attestation}
\author{Dmytro Diikun\thanks{Aspects of the construction described here are the subject of a pending patent application by the author.}\\[2pt]
\small Independent Researcher, Ukraine\\
\small \texttt{diikunapp@gmail.com}\\
\small ORCID: \texttt{0009-0007-5303-3358}\\
\small Also available at Zenodo: \texttt{doi:10.5281/zenodo.22055260}}
\date{August 2026 \quad (preprint)}

\begin{document}
\maketitle
\begin{center}\small \textcopyright~2026 Dmytro Diikun. Licensed under CC BY 4.0.\end{center}

\begin{abstract}
We describe and analyze a protocol for mutual co-signing of a document by two
mobile devices that (i)~communicate only over a one-way, lossy, low-bandwidth
optical channel (an animated on-screen code read by the counterparty's camera),
(ii)~use no intermediary server on the trust path, and (iii)~use no certificate
authority. Trust in each party's public key is instead grounded in a hardware
attestation token produced by the platform secure element, transported in full
over the visual channel by a rateless (fountain) code and cryptographically bound
into the co-signature. The core technical contribution is a \emph{two-stage hash
anchor} that removes the circular signing dependency inherent to interactive
co-signing: the first party commits to the document before the identity of the
second party is known, and the second party's identity is later bound to that
commitment without invalidating the first signature. We give a threat model,
define four security properties---anchor binding, co-signature inseparability,
attestation-bound key provenance, and post-signing tamper evidence---and reduce
them to standard assumptions (collision resistance of $H$ and EUF-CMA security of
the underlying signature scheme), with the secure element modeled as an ideal
signing oracle. We report a working instantiation on iOS/Android using
ECDSA~P-256 in the Secure Enclave/StrongBox, SHA-256, Apple App Attest / Play
Integrity, and an LT-style fountain code, together with an independent
third-party verifier that recomputes all anchors and checks both signatures fully
offline.
\end{abstract}


\section{Introduction}\label{sec:intro}
\paragraph{Problem.}
Two people meet in person, each holding a phone, and want to co-sign a document
so that a third party can later verify both signatures and the fact that each
signing key lives in genuine device hardware. The setting we target has three
hard constraints simultaneously: no network at signing time; no server that the
signed data passes through; and no certificate authority (CA) to vouch for either
public key. The only channel is optical: one phone displays a machine-readable
code, the other reads it with its camera. The channel is one-way per turn, lossy
(dropped frames), and capacity-limited.

\paragraph{Why this is not trivial.} Two difficulties compound.
\begin{enumerate}[leftmargin=1.6em]
\item \emph{Circular dependency.} In naive interactive co-signing, each party
wants to sign ``the final document,'' but the final document includes the other
party's identity and signature, which do not yet exist. Party~$A$ cannot sign
until $B$ is known; $B$ cannot be bound until $A$ has signed. We break this with a
two-stage hash anchor (\S\ref{sec:anchor}).
\item \emph{Key trust without a CA.} With no CA and no network, a receiver cannot
fetch or validate a certificate chain for the sender's key online. We replace
certificate-based trust with \emph{transported hardware attestation}: the full
platform attestation token, which cryptographically ties the public key to key
generation inside a genuine secure element, is carried over the optical channel
itself and bound into the co-signature (\S\ref{sec:attest}). Because a full
attestation object exceeds the capacity of a single static optical code, this is
made practical only by rateless framing (\S\ref{sec:framing}).
\end{enumerate}

\paragraph{Contributions.}
\begin{itemize}[leftmargin=1.6em]
\item A two-stage hash-anchor construction that provably removes the circular
signing dependency while keeping the first signature stable
(\S\ref{sec:anchor},\,\S\ref{sec:anchorbind}).
\item A signature-chaining rule that makes the second signature inseparable from
a \emph{specific} first signature, preventing ``re-pasting'' a signature onto a
different document or a different first signature
(\S\ref{sec:chaining},\,\S\ref{sec:insep}).
\item A CA-less key-provenance mechanism based on transporting the full
hardware-attestation token point-to-point over the visual channel and binding its
hash into the signed anchor, with an offline verification path when the
attestation scheme offers a certificate chain to a manufacturer root and a
detachable third-party path otherwise (\S\ref{sec:attest},\,\S\ref{sec:prov}).
\item An engineering instantiation and a standalone verifier that reproduces all
anchors and both signatures with no network access (\S\ref{sec:impl}).
\end{itemize}

\paragraph{Scope.} This paper is about a cryptographic protocol. It makes no
claim about the legal status of its output in any jurisdiction, and we do not
claim the scheme is a substitute for a qualified or PKI signature; the absence of
a CA is a design goal with explicit trade-offs, discussed in
\S\ref{sec:limits}. Binding a key to a natural person is an application-layer
step outside the offline core and is deliberately excluded from the protocol and
its analysis.

\section{Related Work}\label{sec:related}
Our construction sits at the intersection of five lines of work. For each we
state what it provides and how our setting differs; the contribution is the
\emph{combination} (\S\ref{sec:related-combo}), not any single ingredient.

\paragraph{Multi-signatures and co-signing.}
Multi-signature schemes let a group of signers produce a single \emph{compact}
signature on a \emph{common} message; Bellare and Neven~\cite{BN06} give a scheme
secure in the plain public-key model, and Schnorr-based constructions such as
MuSig~\cite{MPSW19} refine round complexity. These target signature
\emph{aggregation} under a PKI or the plain-public-key model and are interactive
protocols run over a network. Our setting differs on three axes. First, we do not
aggregate: the two parties produce two \emph{distinct} signatures over
\emph{role-separated} messages ($m_A$, $m_B$), and it is their \emph{chaining}
that carries security (\S\ref{sec:chaining}). Second, signing is \emph{ordered
and asymmetric}: the initiator commits before the responder's identity exists,
which our two-stage anchor (\S\ref{sec:anchor}) makes possible. Third, key trust
is established without certificates and without a network at signing time.

\paragraph{Certificate-less and identity-based public-key cryptography.}
Al-Riyami and Paterson~\cite{CLPKC03} introduced certificate-less public-key
cryptography (CL-PKC) to remove certificates while avoiding identity-based key
escrow. CL-PKC still relies on a key-generation centre (KGC) holding a master
secret. Our scheme also dispenses with certificates, but its trust root is
different in kind: the \emph{hardware manufacturer's attestation root}, reached
and checked without any network at signing time (\S\ref{sec:attest}). We do not
replace a CA with a KGC; we replace certificate-based key trust with transported
hardware attestation.

\paragraph{Remote and hardware attestation.}
Platform attestation services---Apple App Attest / DeviceCheck~\cite{AppAttest},
Google Play Integrity~\cite{PlayIntegrity}, Android Key
Attestation~\cite{KeyAttest}, and, in the anonymous setting, TPM-based Direct
Anonymous Attestation~\cite{DAA04}---let a relying party check that a key was
generated inside genuine device hardware. In standard deployment the token is
verified \emph{server-side}. Our use is different: the \emph{full} token is
transported peer-to-peer over the visual channel and \emph{bound into the mutual
signature} (\S\ref{sec:attest}), so a third party can assess key provenance
either fully offline (when the token carries a certificate chain to a pinned
manufacturer root, as with App Attest) or as a detachable step against the
manufacturer's service (as with Play Integrity)---in both cases without a CA.

\paragraph{Rateless codes and animated visual channels.}
Fountain codes---LT~\cite{LubyLT02}, Raptor~\cite{Raptor06}, and
RaptorQ~\cite{RaptorQ}---recover a $k$-symbol payload from any sufficiently large
set of encoded symbols over an erasure channel~\cite{DigitalFountain98}, with no
feedback channel. Combining them with animated on-screen barcodes to move data
screen-to-camera is established in practice (e.g.\ TXQR~\cite{TXQR},
Cimbar~\cite{Cimbar}). In all of these, the animated fountain stream is a
\emph{one-way data pipe}; none binds the transport to a mutual signing act or to
hardware attestation. We use rateless framing for one structural reason: a full
attestation token exceeds the capacity of a single static optical code, so
transporting it point-to-point---and thereby achieving CA-less provenance at
all---is only practical with a rateless code (\S\ref{sec:framing}). We make no
optimality claim for our specific degree rule and defer optimal rateless design
to the LT/Raptor line (\S\ref{sec:framing},\,\S\ref{sec:limits}).

\paragraph{Content provenance and capture-time signing.}
Content-authenticity frameworks, notably C2PA / Content Credentials~\cite{C2PA},
attach a signed manifest asserting an asset's origin and edit history. Two
differences are salient. First, C2PA is \emph{one-sided} capture/edit provenance,
not a mutual co-signature between two independent parties. Second, C2PA hard
bindings hash the \emph{asset bytes}, so any re-encoding invalidates the binding;
we deliberately anchor to the \emph{canonical textual} content and \emph{exclude}
display-file bytes (\S\ref{sec:composite},\,\S\ref{sec:anchor-detach}), so the
display artifact can be regenerated without breaking either signature.

\paragraph{The combination.}\label{sec:related-combo}
To our knowledge, no prior work co-signs a document \emph{mutually},
\emph{offline}, with \emph{certificate-less} key trust established by a
\emph{transported hardware-attestation token} carried over a \emph{rateless
visual channel}, and with a third party able to reproduce both signatures and
both provenance checks without network access. Each ingredient is individually
known; the novelty we claim is their composition and the two mechanisms that make
it work---the two-stage anchor that removes the circular signing dependency, and
the signature chaining that makes the second signature inseparable from a
specific first signature.

\section{Preliminaries and Notation}\label{sec:prelim}
Let $H:\{0,1\}^{\ast}\to\{0,1\}^{256}$ be a collision-resistant hash function
(instantiated as SHA-256). We write $x\cat y$ for concatenation and use two
non-printing control bytes as field delimiters: $\RS$ (ASCII~30, ``record
separator'', \texttt{0x1E}) between top-level fields, and $\US$ (ASCII~31, ``unit
separator'', \texttt{0x1F}) inside a compound field. All strings are UTF-8
encoded before hashing.

A \emph{secure-element signature scheme} $\mathsf{Sig}=(\KGen,\Sign,\Vrfy)$ is a
standard signature scheme (instantiated as ECDSA over NIST~P-256 with SHA-256,
DER-encoded, Base64-serialized~\cite{FIPS186,SEC1}) with the additional property
that $sk$ is generated inside, and never leaves, a hardware secure element; the
signing operation is gated by a user-presence factor (biometric). We model the
secure element as an oracle $\mathcal{O}_{\mathrm{SE}}$ that holds $sk$ and, on
input $m$, returns $\Sign(sk,m)$ only after a presence check; $sk$ is never
exposed. This matches an EUF-CMA signing oracle.

An \emph{attestation scheme} $\Att=(\Gen,\AttVer)$ produces, for a secure-element
public key $pk$ and a challenge $c$, a token $\tau$ such that $\AttVer(pk,c,\tau)$
accepts only if $pk$ was generated inside a genuine secure element of a device of
the attested platform, and $c$ was bound at generation time. We denote by
$\Adv^{\mathrm{snd}}_{\Att}$ the advantage of a PPT adversary against this
soundness property.

\begin{remark}[Production vs.\ development attestation]\label{rem:prod}
The soundness assumption is meaningful only for \emph{production} attestation. A
development-environment token (distinguishable by its attestation metadata, e.g.\
the App Attest environment / \texttt{aaguid}) must not be treated as the full
hardware guarantee. On a device with no secure element the platform path is
unavailable; the implementation then emits a reduced-assurance marker
($\mathit{isReal}=\mathrm{false}$) that is visible to any verifier and carries no
hardware-origin claim (\S\ref{sec:impl}).
\end{remark}

\section{System Model and Threat Model}\label{sec:model}
\subsection{Parties and channel}
Two mobile devices, $D_A$ (initiator) and $D_B$ (responder), each with a secure
element, a camera, and a display. The primary channel is optical: $D_A\to D_B$
(A's animated code read by B), then $D_B\to D_A$ (B's animated code read by A). No
server carries signed data; no CA is contacted.

An optional \emph{online variant} exists in which a relay forwards frames between
remote devices. The relay is explicitly \emph{not on the trust path}: every datum
it forwards is already committed and signed before it reaches the relay, and the
third-party verification recomputes all anchors from the committed fields, so a
malicious relay can at most cause a denial of service, never a change of accepted
content. The same non-network channels (NFC, Bluetooth/BLE, ultrasound, direct
device-to-device links, or machine-readable data embedded in a carrier file) may
substitute for the optical channel; the unifying property is that the
self-contained structure travels without signed data passing through any server.

\subsection{Adversary}\label{sec:adversary}
We consider four adversaries, analyzed separately.
\begin{description}[leftmargin=1.4em,style=nextline]
\item[$\mathcal{A}_{\mathrm{net}}$ --- hostile relay/channel.] Controls any relay
or optical path; may drop, reorder, replay, or inject frames. Goal: alter the
accepted content, or make a verifier accept a document neither party signed.
\item[$\mathcal{A}_{\mathrm{repaste}}$ --- signature re-paster.] A malicious
counterparty or channel observer who has seen valid $(pk_A,\sigma_A)$ and/or B's
signature and tries to transplant a signature onto a different document or a
different first signature.
\item[$\mathcal{A}_{\mathrm{key}}$ --- key-provenance forger.] Tries to make a
verifier accept a key as hardware-backed when it is not (e.g.\ a software key on
an emulator), or to replay another device's token for a fresh session.
\item[$\mathcal{A}_{\mathrm{post}}$ --- post-signing tamperer.] Given the fully
signed document, tries to change any field (amount, party, time, geolocation,
evidence) without detection.
\end{description}

\subsection{Trust assumptions}\label{sec:trust}
$H$ is collision-resistant; $\mathsf{Sig}$ is EUF-CMA-secure; the secure element
does not leak $sk$. The attestation manufacturer root(s) are authentic and are
embedded (pinned) in the verifier out of band. \emph{Out of scope, stated
explicitly:} coercion or duress of a present user; malware with full control of a
rooted/jailbroken device that can itself drive the biometric prompt; and the
legal weight of the output. The implementation carries a \texttt{duressFlag}; it
is a UX signal only, is \emph{not} integrity-protected, and is outside the
cryptographic threat model.

\section{Construction}\label{sec:construction}
We present the scheme bottom-up: canonical encoding (\S\ref{sec:encoding}), the
two-stage anchor (\S\ref{sec:anchor}), signature chaining
(\S\ref{sec:chaining}), attestation transport and binding (\S\ref{sec:attest}),
rateless framing (\S\ref{sec:framing}), and the composite verification hash
(\S\ref{sec:composite}). The analysis assumes only that the canonicalization
$\enc$ is injective and domain-separated by a scheme tag $\ver$
(\S\ref{sec:encoding}); the two serializations we use both satisfy this
(Lemma~\ref{lem:inj}). The worked example in Appendix~\ref{app:worked} is
synthetic and reproducible.

\subsection{Canonical field encoding and separator-injection defense}\label{sec:encoding}
A fixed, ordered field vector $\vec C$ is serialized to a canonical UTF-8 string
$\enc(\vec C\,)$ that is \emph{injective} (distinct vectors yield distinct byte
strings) and \emph{domain-separated} by a scheme tag $\ver$. We use two standard
injective serializations behind the tag. In the \emph{control-byte-escaped} form,
every value passes through an escaping map $f(\cdot)$ that rewrites any embedded
$\RS$ (the top-level separator) to $\US$, and compound sub-fields use $\US$; this
is the form of the composite record (\S\ref{sec:composite}). In the
\emph{length-prefixed} (TLV) form, every value is emitted as
$\langle\text{byte-length}\rangle\,\text{:}\,\langle\text{value}\rangle\,\text{\textbackslash n}$,
so a separator inside a value cannot forge a field boundary regardless of its
content. Both prevent the concatenation collisions in which two distinct field
vectors serialize to the same byte string (Lemma~\ref{lem:inj}); the
length-prefixed form is injective unconditionally, the escaped form on the
admissible domain (Remark~\ref{rem:domain}). We write $\enc(\vec C\,)$ for either
serialization of an ordered content vector $\vec C$.

\subsection{Two-stage hash anchor}\label{sec:anchor}
Let $C=\enc(\vec C\,)$ denote the canonical content string containing the document
terms and \emph{only} party $A$'s identity fields (no $B$ fields, no signatures,
no public keys). Define the first anchor
\[
  h_1 \;=\; H\!\left(\ver\cat \RS \cat C\right).
\]
Crucially, $h_1$ is computable by $D_A$ \emph{before} B's identity is known. When
$B$ joins, define the second anchor as a deterministic function of $h_1$ and $B$'s
identity fields $\mathit{id}_B$ (nickname, avatar index, and an optional
verified-contact hash):
\[
  h_2 \;=\; H\!\big(\ver \cat \RS \cat h_1 \cat \enc(\mathit{id}_B)\big),
\]
computed with the same injective canonicalization as $h_1$.
The scheme tag $\ver$---a short literal fixing the canonicalization, such as the deployed \texttt{v81}/\texttt{v82} identifiers---domain-separates schemes and rules out cross-scheme collisions.
Because $h_2$ is a deterministic function of $h_1$, the first signature
(below) remains valid after $B$ joins---no re-signing round is needed, so the
exchange collapses from three rounds to two optical turns (Figure~\ref{fig:seq}).

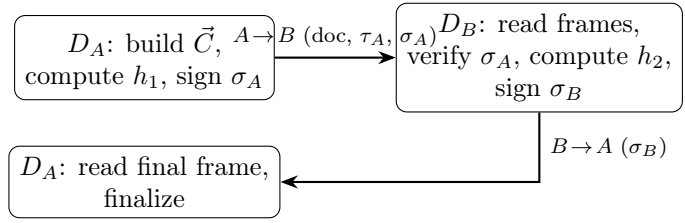
\begin{figure}[t]\centering
\begin{tikzpicture}[node distance=6mm and 16mm,
  box/.style={draw,rounded corners,align=center,font=\small,inner sep=4pt},
  ar/.style={-{Stealth[]},thick}]
  \node[box] (a1) {$D_A$: build $\vec C$,\\ compute $h_1$, sign $\sigma_A$};
  \node[box,right=of a1] (b1) {$D_B$: read frames,\\ verify $\sigma_A$, compute $h_2$,\\ sign $\sigma_B$};
  \node[box,below=of a1] (a2) {$D_A$: read final frame,\\ finalize};
  \draw[ar] (a1) -- node[above,font=\scriptsize]{$A\!\to\!B$ (doc, $\tau_A$, $\sigma_A$)} (b1);
  \draw[ar] (b1) |- node[right,font=\scriptsize,pos=0.25]{$B\!\to\!A$ ($\sigma_B$)} (a2);
\end{tikzpicture}
\caption{Two optical turns. $h_1$ is committed before $B$ exists; $h_2$ binds $B$
without invalidating $\sigma_A$.}\label{fig:seq}
\end{figure}

\subsection{Signatures and chaining}\label{sec:chaining}
With document identifier $\mathit{id}$, the two signed messages are
\[
  m_A = \text{``A:''}\cat \mathit{id}\cat\text{``:''}\cat h_1,\qquad
  m_B = \text{``B:''}\cat \mathit{id}\cat\text{``:''}\cat h_2\cat\text{``:''}\cat
        pk_A\cat\text{``:''}\cat \sigma_A,
\]
and $\sigma_A=\Sign(sk_A,m_A)$, $\sigma_B=\Sign(sk_B,m_B)$, each produced inside
the respective secure element under biometric gating. The domain-separating
prefixes ``A:''/``B:'' prevent cross-role message collisions. Because $m_B$
contains \emph{both} $pk_A$ and $\sigma_A$, the second signature is bound to one
specific, already-existing first signature (Figure~\ref{fig:chain}).

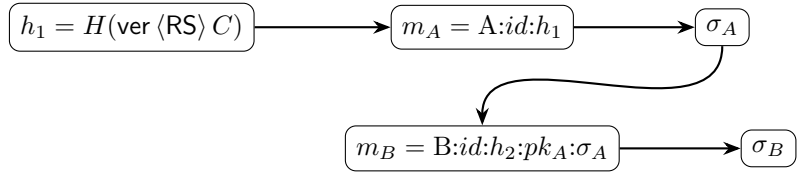
\begin{figure}[t]\centering
\begin{tikzpicture}[
  box/.style={draw,rounded corners,align=center,font=\small,inner sep=4pt},
  ar/.style={-{Stealth[]},thick}]
  \node[box] (ha) {$h_1=H(\ver\,\RS\,C)$};
  \node[box,right=18mm of ha] (ma) {$m_A=\text{A:}id\text{:}h_1$};
  \node[box,right=16mm of ma] (sa) {$\sigma_A$};
  \node[box,below=10mm of ma] (mb) {$m_B=\text{B:}id\text{:}h_2\text{:}pk_A\text{:}\sigma_A$};
  \node[box,right=16mm of mb] (sb) {$\sigma_B$};
  \draw[-{Stealth[]},thick] (ha)--(ma);
  \draw[-{Stealth[]},thick] (ma)--(sa);
  \draw[-{Stealth[]},thick] (sa) to[out=-90,in=90] (mb.north);
  \draw[-{Stealth[]},thick] (mb)--(sb);
\end{tikzpicture}
\caption{Signature chaining: $\sigma_A$ and $pk_A$ are inputs to $m_B$, so
$\sigma_B$ is inseparable from this specific $\sigma_A$.}\label{fig:chain}
\end{figure}

\subsection{Transported hardware attestation and key provenance}\label{sec:attest}
$D_A$ obtains a full attestation token $\tau_A$ binding $pk_A$ to secure-element
key generation, with a fresh single-use challenge $c$ committing to the key: the
platform attests a client-data hash $\mathit{cdh}=H(c\cat H(pk_A))$. The
\emph{full} $\tau_A$ (not merely a boolean ``is-real'') is placed in the transfer
structure; only its context-separated digest
$H(\text{``}\mathsf{ctx\_device\_A}\text{:''}\cat\tau_A)$ enters the signed anchor
set (\S\ref{sec:composite}). Verification has two modes:
(i)~\emph{offline chain mode}---the token carries a certificate chain to a
manufacturer root pinned in the verifier (e.g.\ Apple App Attest on iOS), checked
locally with no network; and (ii)~\emph{detachable service mode}---the token is
verifiable only against the manufacturer's service (e.g.\ Google Play Integrity on
Android), so the offline phase transports and binds $\tau$ and the hardware-origin
check completes as a detachable third-party step. Soundness and no-replay are
proved in \S\ref{sec:prov}.

\subsection{Rateless framing of the transfer structure}\label{sec:framing}
The self-contained transfer structure (content fields, $h_1$, $\sigma_A$, $pk_A$,
full $\tau_A$, scheme id, session id) is compressed and split into $k$ base
fragments ($\le\!675$ bytes each). Base frames carry one fragment; fountain frames
carry the XOR of a pseudo-random subset chosen by a linear congruential generator
(multiplier $1664525$, increment $1013904223$, modulus $2^{32}$, seeded by the
frame sequence number), with redundancy factor $3$ and degree $(\mathrm{seq}\bmod
k)+1$. The decoder peels fragments and confirms completion via a short checksum
(a prefix of $H(\text{payload})$). Because the code is rateless, no back-channel is
needed to request lost frames (Figure~\ref{fig:fountain}). We emphasize that the
specific LCG-XOR degree rule is an \emph{engineering choice}, not an optimized
construction: we claim only that it decodes reliably at the payload sizes in use.
Optimal rateless design (Robust-Soliton LT, Raptor/RaptorQ) is orthogonal and
applies unchanged as a drop-in replacement.

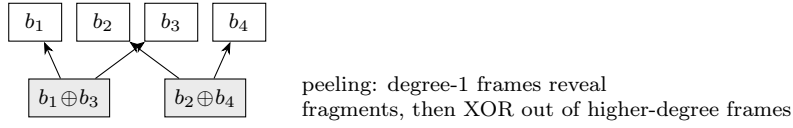
\begin{figure}[t]\centering
\begin{tikzpicture}[
  frag/.style={draw,minimum width=7mm,minimum height=5mm,font=\scriptsize},
  ar/.style={-{Stealth[]}}]
  \foreach \i in {1,2,3,4} { \node[frag] (f\i) at (\i*0.9,1) {$b_\i$}; }
  \node[frag,fill=black!8] (x1) at (1.35,0) {$b_1\!\oplus\!b_3$};
  \node[frag,fill=black!8] (x2) at (3.15,0) {$b_2\!\oplus\!b_4$};
  \draw[ar] (x1) -- (f1); \draw[ar] (x1) -- (f3);
  \draw[ar] (x2) -- (f2); \draw[ar] (x2) -- (f4);
  \node[font=\scriptsize,align=left,right=6mm of x2] {peeling: degree-1 frames reveal\\ fragments, then XOR out of higher-degree frames};
\end{tikzpicture}
\caption{Fountain framing and peeling decode over the lossy optical
channel.}\label{fig:fountain}
\end{figure}

\subsection{Composite verification hash}\label{sec:composite}
For archival and third-party checking, a composite hash $\mathit{finalHash}$ is
computed over an ordered 39-field record in the control-byte-escaped form of
\S\ref{sec:encoding} (both parties' public keys,
signatures, biometric-backing flags, device-token hashes, timestamps,
geolocation, evidence hashes, identity hashes, and $h_1$), with the $B$-block
included only when $\sigma_B\neq\varepsilon$. Every value is passed through the
escaping map $f$; the separator $\RS$ is written after every field. The exact
field order is reproduced verbatim by the reference verifier and listed in
Appendix~\ref{app:fields}.

\paragraph{Time-source field.}\label{sec:time}
One canonical field records the time source as a short string of the form
$\langle\text{ISO-8601}\rangle\,|\,\langle\text{source-id}\rangle$, where the
source identifier ranks, in decreasing order of trust, a relay-authenticated
timestamp, an NTP-synchronized time, and the local device clock (marked as such).
The security-relevant fact is only that this string is included verbatim in the
canonical text and therefore covered by $h_1$: by Theorem~\ref{thm:tamper}, any
retroactive substitution of the source (e.g.\ silently downgrading to the local
clock, or claiming a higher-trust source) changes the recomputed anchor and is
detected. We make no non-repudiation claim for the time source itself; a
higher-trust timestamp is a strengthening input, not part of the offline
security core.

\paragraph{Anchor detached from display bytes.}\label{sec:anchor-detach}
The cryptographic anchor is the canonical \emph{text} ($h_1$, $\mathit{finalHash}$),
\emph{not} the bytes of any display file (e.g.\ PDF). Non-deterministic rendering
bytes are deliberately excluded from the canonical representation, the signed
messages, and the transfer structure, so the display artifact can be regenerated
at will without breaking any signature.

\paragraph{Generality and variants.}\label{sec:variants}
The construction is stated for two parties over an optical channel, but the same
mechanisms apply to a family of instantiations. \emph{Number of parties:} the
two-party case is the base of a sequential $N$-party co-signing ($N\ge 2$), where
the $i$-th signature is taken over an anchor depending on the $i$-th party's
identity fields and chained to the public keys and signatures of all previous
parties; the one-sided ($N{=}1$) case is a special case. \emph{Channel:} any
non-network carrier of the self-contained structure works---animated optical code,
NFC, Bluetooth/BLE, ultrasound, a direct device-to-device link, or machine-readable
data embedded in a carrier file---the unifying property being that signed data
never passes through a server. \emph{Primitives:} $H$ may be any collision-resistant
hash (SHA-256, SHA-512, SHA-3, BLAKE2/BLAKE3); $\mathsf{Sig}$ any EUF-CMA scheme
with a secure-element key (ECDSA P-256/secp256k1, Ed25519, RSA); the secure element
any of Secure Enclave, StrongBox, TEE, TPM, or smart card; and the attestation any
scheme exposing key-in-hardware provenance (App Attest/DeviceCheck, Play Integrity,
Android Key Attestation, FIDO, TPM quote). \emph{Framing:} the rateless code may be
LT, Raptor/RaptorQ, an online code, or a fixed multi-frame scheme recoverable from
a subset. A trusted timestamp (e.g.\ RFC~3161) is an optional, detachable
strengthening, not part of the offline core. All statements below depend only on
the abstract properties (collision resistance of $H$, EUF-CMA of $\mathsf{Sig}$,
soundness of $\Att$, injectivity of $\enc$), so they carry over to every
instantiation in this family unchanged.

\section{Security Analysis}\label{sec:security}
We define and prove four properties. Throughout, the secure element is modeled as
an ideal signing oracle (\S\ref{sec:prelim}), and ``accept'' means the
third-party verifier's check passes:
\[
  \Vrfy^{\ast}(R)=1 \iff \Vrfy(pk_A,m_A,\sigma_A)=1 \;\wedge\;
  \Vrfy(pk_B,m_B,\sigma_B)=1,
\]
where $R=(\mathit{id},C,\mathit{id}_B,pk_A,\sigma_A,pk_B,\sigma_B)$ and $m_A,m_B$
are derived as in \S\ref{sec:chaining}. Attestation checks are orthogonal and
treated in \S\ref{sec:prov}.

\subsection{Encoding injectivity and anchor binding}\label{sec:anchorbind}

\begin{lemma}[Encoding injectivity]\label{lem:inj}
Fix a field schema of arity $n$ (an ordered list of $n$ slots, each either a
simple slot or a compound slot of two sub-fields). Let $\enc$ apply the escaping
map $f$ to every field value (replacing every $\RS$ byte inside a value by $\US$),
join the $n$ resulting tokens with $\RS$, and join the two sub-fields of a
compound slot with $\US$; the byte $\RS$ appears in the output only as a top-level
separator. Then, restricted to inputs whose sub-field values contain no $\US$ byte
in a position that would forge a sub-boundary, $\enc$ is injective: for any two
field vectors $\vec v\neq\vec w$ of the same schema, $\enc(\vec v)\neq\enc(\vec
w)$.
\end{lemma}

\begin{proof}
The arity $n$ and the type of each slot are fixed and public, so a decoder knows
where boundaries are expected; we show $\enc$ is decodable, which implies
injectivity. Given $s=\enc(\vec v)$, split $s$ on $\RS$. By construction $\RS$
occurs only between top-level tokens (inside every value $f$ has rewritten each
$\RS$ to $\US$), so the split yields exactly the $n$ tokens $f(v_1),\dots,f(v_n)$
in order. For a compound slot, split its token on $\US$; the schema fixes the slot
as compound and forbids a spurious $\US$, recovering the ordered pair. Thus
$\enc(\vec v)$ determines $(f(v_1),\dots,f(v_n))$ componentwise. If $\vec
v\neq\vec w$ differ in component $i$ with $f(v_i)\neq f(w_i)$, the encodings
differ. The remaining case $v_i\neq w_i$ but $f(v_i)=f(w_i)$ requires values
differing only in $\RS$-vs-$\US$ bytes; on the admissible domain $f$ is the
identity there, so $f(v_i)=f(w_i)\Rightarrow v_i=w_i$, a contradiction. Hence
$\enc(\vec v)\neq\enc(\vec w)$.
\end{proof}

\begin{remark}[Domain caveat, stated honestly]\label{rem:domain}
Injectivity holds on the admissible domain, i.e.\ where field values do not
themselves contain a raw $\RS$ that a caller expected to survive as data. In
practice the hashed fields are UUIDs, base64/hex digests, ISO-8601 timestamps, and
short human strings, none of which legitimately contain the control bytes $\RS$
(\texttt{0x1E}) or $\US$ (\texttt{0x1F}); the scheme tag $\ver$ additionally
domain-separates the scheme. The length-prefixed (TLV) serialization is a standard
technique that removes the caveat entirely: it is injective unconditionally. When
it is used for the content anchors, Lemma~\ref{lem:inj} holds without restriction;
the escaped form is retained for the composite record and for verifying documents
signed under it. Both coexist behind the scheme tag.
\end{remark}

\begin{definition}[Anchor binding]\label{def:bind}
An adversary wins if it outputs either (i)~two distinct A-content vectors $\vec
C\neq\vec C'$ with $h_1(\vec C\,)=h_1(\vec C'\,)$, or (ii)~two distinct bindings
$(h_1,\mathit{id}_B)\neq(h_1',\mathit{id}_B')$ with $h_2=h_2'$. Its advantage is
$\Adv^{\mathrm{bind}}$.
\end{definition}

\begin{theorem}[Anchor binding]\label{thm:bind}
If $H$ is collision-resistant and $\enc$ is injective on the admissible domain
(Lemma~\ref{lem:inj}), then $\Adv^{\mathrm{bind}}\le\Adv^{\mathrm{cr}}_{H}$.
\end{theorem}

\begin{proof}
Case (i): $h_1(\vec C\,)=h_1(\vec C'\,)$ means $H(\ver\cat\RS\cat\enc(\vec
C\,))=H(\ver\cat\RS\cat\enc(\vec C'\,))$. Since $\vec C\neq\vec C'$,
Lemma~\ref{lem:inj} gives $\enc(\vec C\,)\neq\enc(\vec C'\,)$, so the two hash
inputs are distinct: a collision of $H$. Case (ii): $h_2=h_2'$ means
$H(h_1\cat\RS\cat e)=H(h_1'\cat\RS\cat e')$ where $e,e'$ are the (escaped,
$\RS$-separated) encodings of $\mathit{id}_B,\mathit{id}_B'$. If
$(h_1,\mathit{id}_B)\neq(h_1',\mathit{id}_B')$ then the fixed-length $h_1$ prefix
and the injective encoding of the identity fields make the two inputs distinct: a
collision. In either case the reduction outputs the colliding pair.
\end{proof}

\subsection{Co-signature inseparability}\label{sec:insep}
We formalize that $\sigma_B$ cannot be detached from the specific first signature
it endorses and re-attached to a different one.

\begin{definition}[Co-signature inseparability]\label{def:insep}
The challenger runs $(pk_B,sk_B)\leftarrow\KGen(1^\lambda)$ inside an ideal secure
element, gives $pk_B$ to $\mathcal{A}$, and keeps a set $Q$. On a query
$(\mathit{id},C,\mathit{id}_B,pk_A,\sigma_A)$ with $\Vrfy(pk_A,m_A,\sigma_A)=1$,
the oracle $\mathcal{O}_B$ computes $h_2$, forms $m_B$, returns
$\sigma_B\leftarrow\Sign(sk_B,m_B)$, and records
$(\mathit{id},h_2,pk_A,\sigma_A)\in Q$. $\mathcal{A}$ may generate its own A-side
keys and signatures. It outputs
$R^{\ast}=(\mathit{id}^{\ast},C^{\ast},\mathit{id}_B^{\ast},pk_A^{\ast},
\sigma_A^{\ast},pk_B,\sigma_B^{\ast})$; let
$h_2^{\ast}=h_2(h_1(C^{\ast}),\mathit{id}_B^{\ast})$. $\mathcal{A}$ wins if
$\Vrfy^{\ast}(R^{\ast})=1$ and
$(\mathit{id}^{\ast},h_2^{\ast},pk_A^{\ast},\sigma_A^{\ast})\notin Q$.
\end{definition}

\begin{theorem}[Inseparability]\label{thm:insep}
If $\mathsf{Sig}$ is EUF-CMA-secure then for every PPT $\mathcal{A}$,
$\Adv^{\mathrm{insep}}_{\mathcal{A}}\le\Adv^{\mathrm{euf\text{-}cma}}_{\mathcal{B},\mathsf{Sig}}$.
\end{theorem}

\begin{proof}
$\mathcal{B}$ receives a challenge key $pk^{\ast}$ with signing oracle
$\mathcal{O}_{\mathrm{Sign}}$ and sets $pk_B:=pk^{\ast}$ (a perfect simulation,
since $sk_B$ is used only through the ideal secure element). It answers each
$\mathcal{O}_B$ query by forming $m_B$ and calling
$\sigma_B\leftarrow\mathcal{O}_{\mathrm{Sign}}(m_B)$, logging $m_B$ and recording
$Q$. On a winning $R^{\ast}$, $\mathcal{B}$ outputs $(m_B^{\ast},\sigma_B^{\ast})$
with $m_B^{\ast}=\text{``B:''}\cat\mathit{id}^{\ast}\cat\text{``:''}\cat
h_2^{\ast}\cat\text{``:''}\cat pk_A^{\ast}\cat\text{``:''}\cat\sigma_A^{\ast}$.
Validity: $\Vrfy(pk^{\ast},m_B^{\ast},\sigma_B^{\ast})=1$ by assumption. It
remains to show $m_B^{\ast}$ was never queried. The template
``B:''$\,\mathit{id}\,$``:''$\,h_2\,$``:''$\,pk_A\,$``:''$\,\sigma_A$ is uniquely
decodable: $\mathit{id}$ is a UUID containing no ``:'', $h_2$ is a fixed-length
hex digest, $pk_A$ occupies two known-length base64 tokens, and $\sigma_A$ is the
final base64 token; base64 and hex alphabets exclude ``:''. Hence
$m_B^{\ast}$ determines $(\mathit{id}^{\ast},h_2^{\ast},pk_A^{\ast},
\sigma_A^{\ast})$ uniquely; if it equalled a queried $m_B$ then the corresponding
recorded tuple would equal
$(\mathit{id}^{\ast},h_2^{\ast},pk_A^{\ast},\sigma_A^{\ast})\in Q$, contradicting
the winning condition. So $m_B^{\ast}$ is fresh and $(m_B^{\ast},\sigma_B^{\ast})$
is a valid forgery.
\end{proof}

\begin{remark}[Why the chaining fields are load-bearing]\label{rem:chain}
The extraction used only that $m_B$ commits injectively to
$(\mathit{id},h_2,pk_A,\sigma_A)$. Dropping $pk_A$ or $\sigma_A$ from $m_B$ would
let $\mathcal{A}$ present the same $\sigma_B$ under a different $\sigma_A^{\ast}$
without leaving $Q$, and the contradiction step fails. The chain $pk_A,\sigma_A$
inside $m_B$ is exactly what buys inseparability.
\end{remark}

\subsection{Attestation-bound key provenance}\label{sec:prov}
Recall $\Att=(\Gen,\AttVer)$ from \S\ref{sec:prelim} and the challenge binding
$\mathit{cdh}=H(c\cat H(pk_A))$ with $c$ fresh, single-use, and session-scoped
(\S\ref{sec:attest}).

\begin{definition}[Provenance soundness]\label{def:prov}
A PPT $\mathcal{A}_{\mathrm{key}}$ outputs $(pk^{\ast},c^{\ast},\tau^{\ast})$. Let
$\mathsf{Genuine}(pk^{\ast})$ be the event that $pk^{\ast}$ was generated inside a
genuine secure element under $\Att$'s soundness model, and $\mathcal{C}$ the set
of issued (fresh, single-use) challenges. $\mathcal{A}_{\mathrm{key}}$ wins if
$\AttVer(pk^{\ast},c^{\ast},\tau^{\ast})=1$ and either
(a)~$\neg\mathsf{Genuine}(pk^{\ast})$, or (b)~$c^{\ast}\notin\mathcal{C}$ or the
accepted $\mathit{cdh}^{\ast}\neq H(c^{\ast}\cat H(pk^{\ast}))$.
\end{definition}

\begin{proposition}[Provenance / no replay]\label{prop:prov}
If $\Att$ is sound and $H$ is collision-resistant, then
$\Adv^{\mathrm{prov}}_{\mathcal{A}}\le\Adv^{\mathrm{snd}}_{\Att}+\Adv^{\mathrm{cr}}_{H}$.
\end{proposition}

\begin{proof}
Suppose $\mathcal{A}_{\mathrm{key}}$ wins, so
$\AttVer(pk^{\ast},c^{\ast},\tau^{\ast})=1$. \emph{Case (a):}
$\neg\mathsf{Genuine}(pk^{\ast})$ is exactly a break of $\Att$-soundness; the
reduction outputs $(pk^{\ast},c^{\ast},\tau^{\ast})$, bounded by
$\Adv^{\mathrm{snd}}_{\Att}$. \emph{Case (b):} by $\Att$-soundness the bound
client-data hash is the value fixed at generation. If
$\mathit{cdh}^{\ast}=H(c\cat H(pk))$ for some \emph{issued} $(c,pk)\neq(c^{\ast},
pk^{\ast})$, then $H(c\cat H(pk))=H(c^{\ast}\cat H(pk^{\ast}))$ on distinct
preimages---an $H$-collision, bounded by $\Adv^{\mathrm{cr}}_{H}$. Otherwise no
issued pair explains $\mathit{cdh}^{\ast}$, so $\AttVer$ accepted a challenge
never bound to any issued attestation---again a soundness break. Freshness and
single use of $c$ rule out cross-session reuse. Summing gives the bound.
\end{proof}

\begin{remark}[Trust root]\label{rem:root}
Proposition~\ref{prop:prov} reduces provenance to attestation soundness and $H$;
it does \emph{not} eliminate a trusted party. It replaces CA trust with the
hardware-manufacturer attestation root, and a compromise of that root breaks
provenance---a deliberate trade (\S\ref{sec:limits}). Attestation binds the key to
genuine device hardware, not to a natural person.
\end{remark}

\subsection{Post-signing tamper evidence}\label{sec:tamper}
Write $\vec C$ for the vector of A-content fields, $C=\enc(\vec C\,)$,
$h_1=H(\ver\cat\RS\cat C)$, and $m_A=\text{``A:''}\cat\mathit{id}\cat
\text{``:''}\cat h_1$.

\begin{definition}[Content tamper resistance]\label{def:tamper}
The challenger generates $(pk_A,sk_A)$ in an ideal secure element. On a content
vector $\vec C$ (with id $\mathit{id}$) the oracle returns
$\sigma_A\leftarrow\Sign(sk_A,m_A)$ and records $\vec C\in Q_A$. $\mathcal{A}$
outputs $(\vec C^{\ast},\sigma_A^{\ast})$ and wins if
$\Vrfy(pk_A,m_A^{\ast},\sigma_A^{\ast})=1$ and $\vec C^{\ast}\notin Q_A$.
\end{definition}

\begin{theorem}[Tamper evidence]\label{thm:tamper}
If $\mathsf{Sig}$ is EUF-CMA-secure and $H$ is collision-resistant, then
$\Adv^{\mathrm{tamper}}_{\mathcal{A}}\le
\Adv^{\mathrm{euf\text{-}cma}}_{\mathcal{B}_1,\mathsf{Sig}}+
\Adv^{\mathrm{cr}}_{\mathcal{B}_2,H}$.
\end{theorem}

\begin{proof}
Let $\mathcal{A}$ win with $(\vec C^{\ast},\sigma_A^{\ast})$, $\vec
C^{\ast}\notin Q_A$. Split on $m_A^{\ast}$. \emph{Event $E_1$: $m_A^{\ast}$ was
never returned by the oracle.} $\mathcal{B}_1$ sets $pk_A:=pk^{\ast}$, answers
signing queries via its oracle, and outputs $(m_A^{\ast},\sigma_A^{\ast})$, a
fresh valid EUF-CMA forgery; thus $\Pr[\mathcal{A}\text{ wins}\wedge
E_1]\le\Adv^{\mathrm{euf\text{-}cma}}_{\mathcal{B}_1}$. \emph{Event $E_2$:
$m_A^{\ast}$ equals some queried $m_A$ for $\vec C\in Q_A$.} The template
``A:''$\,\mathit{id}\,$``:''$\,h_1$ is uniquely decodable, so $m_A^{\ast}=m_A$
forces $h_1^{\ast}=h_1$, i.e.\ $H(\ver\cat\RS\cat C^{\ast})=H(\ver
\cat\RS\cat C)$. Since $\vec C^{\ast}\neq\vec C$ (one is in $Q_A$, the other is
not), Lemma~\ref{lem:inj} gives $C^{\ast}\neq C$, so the inputs differ: a
collision. $\mathcal{B}_2$ (holding $sk_A$, simulating exactly) outputs the
colliding pair; thus $\Pr[\mathcal{A}\text{ wins}\wedge
E_2]\le\Adv^{\mathrm{cr}}_{\mathcal{B}_2}$. $E_1,E_2$ are exhaustive and disjoint.
\end{proof}

\begin{remark}[Coverage]\label{rem:cover}
$\vec C$ comprises the document terms and A-identity fields of
\S\ref{sec:encoding}; every one is inside $h_1$ and hence under $\sigma_A$. The
39-field $\mathit{finalHash}$ record (Appendix~\ref{app:fields}) additionally
binds both signatures, both public keys, the biometric-backing flags, the
device-token hashes, and $h_1$, so the argument extends to the $B$-block once
$\sigma_B$ is present (via Theorem~\ref{thm:insep} for the chain and this theorem
for the content). Fields deliberately excluded---the non-deterministic rendering
bytes of any display file (\S\ref{sec:anchor-detach})---are, by design, not
covered.
\end{remark}

\section{Implementation}\label{sec:impl}
\begin{itemize}[leftmargin=1.4em]
\item \emph{Primitives.} ECDSA~P-256 (\texttt{prime256v1}), SHA-256; public keys
serialized as $\mathrm{B64}(X){:}\mathrm{B64}(Y)$; signatures DER + Base64. Secure
element: Apple Secure Enclave / Android StrongBox-Keystore. The signing path
admits \emph{no} software key: a device without a secure element yields an error,
not a silent software fallback, and every signing site enforces a hardware gate
($\mathit{backingType}\neq\text{HW}\Rightarrow$ abort). A strict DER decoder
rejects non-canonical encodings and out-of-range $r,s$; low-$S$ signatures are
normalized (Secure Enclave emits high-$S$), and the public point is checked to lie
on P-256.
\item \emph{Attestation.} Apple App Attest (iOS) and Google Play Integrity
(Android); production environment only ($\mathit{allow\_dev}=\text{false}$), cf.\
Remark~\ref{rem:prod}. App Attest tokens (iOS) admit offline chain-to-Apple-root
verification (mode~i) and Play Integrity (Android) is manufacturer-service-verified
(mode~ii). The reference deployment verifies both server-side at onboarding, via a
challenge $c$ with client-data hash $\mathit{cdh}=H(c\,\|\,H(pk))$, and records a
server-signed (ECDSA~P-256) attestation that any party can re-check with the
server's public key; an App Attest token additionally remains offline-verifiable by
a third party that pins Apple's root.
\item \emph{Transport.} Animated optical code with the LT-style fountain code of
\S\ref{sec:framing}; NFC, BLE, ultrasound, and file-embedded channels are
supported alternatives.
\item \emph{Verifier.} A standalone verifier recomputes $h_1$, $h_2$, and
$\mathit{finalHash}$ and checks $\sigma_A,\sigma_B$ with no network access. A
public instance that performs the recomputation locally in the browser is
available at \url{https://pakto.pro/en/verify}. The protocol core is implemented
twice (Dart on the client, Python in the independent verifier), and byte-for-byte
agreement of all digests across the two implementations was confirmed, so
verification does not depend on a single implementation.
\end{itemize}

\section{Limitations and Trade-offs}\label{sec:limits}
\begin{itemize}[leftmargin=1.4em]
\item \emph{No CA is a trade, not a free win.} Key trust rests entirely on the
attestation manufacturer root and its soundness; a break there breaks provenance.
We do not remove a trusted party---we relocate it from a CA to the hardware
manufacturer.
\item \emph{Provenance is per-device, not per-natural-person.} Attestation binds a
key to genuine hardware, not to an individual. Binding a key to a real-world
identity is an application-layer step outside the offline core and is not analyzed
here; it must not be read as a cryptographic guarantee of the core.
\item \emph{Attestation strength is platform-dependent.} StrongBox, TEE, and
software backings differ in assurance, and Android is fragmented; production vs.\
development attestation matters (Remark~\ref{rem:prod}).
\item \emph{Injectivity is domain-restricted.} Lemma~\ref{lem:inj} holds on the
admissible domain; a fully unconditional statement requires TLV encoding
(Remark~\ref{rem:domain}).
\item \emph{Fountain parameters are an engineering choice.} No optimality is
claimed for the custom degree rule (\S\ref{sec:framing}).
\item \emph{Out of scope.} Duress/coercion of a present user; malware on a rooted
device able to drive the biometric prompt; the \texttt{duressFlag} is not
integrity-protected.
\item \emph{No external audit yet.} The implementation has not undergone an
independent third-party security audit; we regard such an audit as a prerequisite
to any high-assurance deployment.
\end{itemize}

\section{Conclusion}\label{sec:conclusion}
We showed that mutual co-signing survives the removal of all standard
infrastructure---no network, no server on the trust path, and no CA---provided one
handles the circular signing dependency with a two-stage hash anchor, makes the
second signature inseparable from a specific first signature by chaining, and
grounds key trust in a hardware-attestation token transported over the channel
itself. The four properties reduce to collision resistance of $H$ and EUF-CMA
security of the signature scheme, and an independent verifier reproduces every
anchor and both signatures offline. The construction does not replace PKI; it
maps out how much verifiability remains when the infrastructure is taken away.

\appendix
\section{Worked Example (redacted)}\label{app:worked}
All values below use \emph{synthetic, non-personal} inputs and can be recomputed
with the reference verifier. The content anchor uses the length-prefixed
serialization (\S\ref{sec:encoding}): each field is emitted as
$\langle\text{byte-length}\rangle\text{:}\langle\text{value}\rangle\text{\textbackslash n}$,
after the tag \texttt{PAKTO-v82\textbackslash n} and a domain label.
\begin{quote}\small\ttfamily
preimage (head): PAKTO-v82\textbackslash n 11:contentHash\textbackslash n\\
\ \ 36:11111111-2222-3333-4444-555555555555\textbackslash n 7:General\textbackslash n \ldots\ 5:alice\textbackslash n 1:1\textbackslash n\\[3pt]
contentHash\ \ \ \ \ = 8271868f458a447f70bd54eda563b7347c0ea330fd1a77e24535fd719332384e\\
contentHashFull = 7ca362ee0b7663f64480b746d4835044f815c2be7e2abed681ed0de84cb4120b\\[3pt]
m\_A = A:11111111-2222-3333-4444-555555555555:8271868f\ldots\\
m\_B = B:11111111-2222-3333-4444-555555555555:7ca362ee\ldots:<pkA>:<sigA>
\end{quote}
Here $\sigma_A=\Sign(sk_A,m_A)$ and $\sigma_B=\Sign(sk_B,m_B)$ are produced in the
respective secure elements. Only device-token \emph{hashes} enter the composite
record (Appendix~\ref{app:fields}); the full attestation tokens travel in the
transfer structure (\S\ref{sec:framing}).

\section{Composite \texttt{finalHash} field order (escaped record)}\label{app:fields}
The composite hash is
$\mathit{finalHash}=H\!\big(\mathrm{utf8}(F_0\,\RS\,F_1\,\RS\cdots\RS\,F_{38}\,\RS)\big)$,
i.e.\ $\RS$ (\texttt{0x1E}) is written after every field including the last. Each
$F_i$ is passed through the escaping map $f$ of \S\ref{sec:encoding} unless noted.
Compound profile fields use $\US$ (\texttt{0x1F}) between nickname and avatar
index. The $B$-block (indices 28--35) is included iff $\sigma_B\neq\varepsilon$;
indices 36--38 always follow. This composite record uses the control-byte-escaped
serialization and carries its own scheme tag in field~0; the \emph{signed} content
anchors $h_1,h_2$ use the length-prefixed (TLV) serialization of
Appendix~\ref{app:worked}. The example column is illustrative and abbreviated: it
fixes the \emph{field order}, not a numeric artifact reproduced elsewhere in this
paper.

{\small
\begin{longtable}{@{}rlll@{}}
\toprule
\# & Field & Fallback / rule & Example (abbrev.)\\
\midrule
\endhead
0  & schema\_prefix & literal scheme tag & \texttt{<scheme-tag>}\\
1  & contractId & --- & \texttt{11111111-...}\\
2  & contractType & --- & \\
3  & contractCategory & --- & \\
4  & title & --- & \textit{(value)}\\
5  & details & --- & \textit{(value)}\\
6  & amount & (empty $\to$ 0 bytes) & (empty)\\
7  & currency & --- & \texttt{EUR}\\
8  & contractLanguage & --- & \texttt{en}\\
9  & jurisdiction & (empty $\to$ \texttt{nojuris}) & \\
10 & createdAt & ISO-8601 UTC & \texttt{2026-01-01T00:00:00Z}\\
11 & tsaTime & ISO-8601 UTC & \texttt{2026-01-01T00:00:00Z}\\
12 & timeSource & --- & \texttt{LOCAL:WARNING}\\
13 & publicKeyA & \texttt{B64(X):B64(Y)} & \textit{(base64:base64)}\\
14 & signatureA & DER+B64 & \textit{(DER+base64)}\\
15 & biometricBackedA & (empty $\to$ \texttt{SW}) & \texttt{HW}\\
16 & deviceTokenHashA & $H(\text{ctx\_device\_A:}\cat\tau_A)$; (empty $\to$ \texttt{nodeviceA}) & \textit{(hash)}\\
17 & geoHash & (empty $\to$ \texttt{nogeo}) & \textit{(redacted)}\\
18 & geoCoordinates & (empty $\to$ \texttt{nocoords}) & (test)\\
19 & tgPhoneHashA & (empty $\to$ \texttt{notg}) & \texttt{notg}\\
20 & emailA & (empty $\to$ \texttt{noemail}) & \texttt{noemail}\\
21 & identityRnokppHashA & (empty $\to$ \texttt{noidentity}) & \textit{(redacted)}\\
22 & identityQualifiedA & \texttt{qualified}/\texttt{notqualified} & \texttt{qualified}\\
23 & tgPhoneHashB & (empty $\to$ \texttt{notg}) & \texttt{notg}\\
24 & evidenceHashesJson & (empty $\to$ \texttt{nophoto}) & \texttt{["bc118af04b..."]}\\
25 & photoHashA & (empty $\to$ \texttt{nophoto}) & \texttt{nophoto}\\
26 & photoHashB & (empty $\to$ \texttt{nophoto}) & \texttt{nophoto}\\
27 & tsaHashA & (empty $\to$ \texttt{notsa}) & \texttt{notsa}\\
\midrule
\multicolumn{4}{@{}l}{\emph{B-block, indices 28--35, present iff $\sigma_B\neq\varepsilon$:}}\\
28 & publicKeyB & \texttt{B64(X):B64(Y)} & \textit{(base64:base64)}\\
29 & signatureB & DER+B64 & \textit{(DER+base64)}\\
30 & biometricBackedB & (empty $\to$ \texttt{SW}) & \texttt{HW}\\
31 & deviceTokenHashB & $H(\text{ctx\_device\_B:}\cat\tau_B)$; (empty $\to$ \texttt{nodeviceB}) & \textit{(hash)}\\
32 & tsaTokenHashB & (empty $\to$ \texttt{notsaB}) & \texttt{notsaB}\\
33 & emailB & (empty $\to$ \texttt{noemail}) & \texttt{noemail}\\
34 & identityRnokppHashB & (empty $\to$ \texttt{noidentity}) & \texttt{noidentity}\\
35 & identityQualifiedB & \texttt{qualified}/\texttt{notqualified} & \texttt{notqualified}\\
\midrule
\multicolumn{4}{@{}l}{\emph{always present:}}\\
36 & contentHash ($=h_1$) & (empty $\to$ \texttt{nocontent}) & \texttt{8271868f...}\\
37 & profileA & $\mathrm{nick}\,\US\,\mathrm{avatar}$; (empty $\to$ \texttt{noprofile}) & \texttt{alice <US> 1}\\
38 & profileB & $\mathrm{nick}\,\US\,\mathrm{avatar}$; (empty $\to$ \texttt{noprofile}) & \texttt{<US> 0}\\
\bottomrule
\end{longtable}}
\noindent For a complete document the reference verifier recomputes
$\mathit{finalHash}$ from these fields and compares it to the stored value.

\end{document}